\documentclass[a4paper,USenglish,runningheads]{llncs}
\usepackage{wrapfig}
\usepackage{microtype}
\usepackage{thm-restate}
\usepackage{ntheorem}
\theoremstyle{plain}
\theorembodyfont{}
\theoremsymbol{}
\theoremprework{}
\theorempostwork{}
\theoremseparator{.}

\newtheorem{observation}{Observation}
\renewtheorem{problem}{Open problem}
\renewtheorem{theorem}{Theorem}
\renewtheorem{lemma}{Lemma}
\renewtheorem{claim}{Claim}
\renewtheorem{property}{Property}

\usepackage{graphicx}
\usepackage{paralist}
\usepackage{xspace}
\usepackage{amsmath,amssymb,gensymb}
\usepackage[colorlinks=true]{hyperref}
\usepackage[color=yellow]{todonotes}
\usepackage[caption=false]{subfig} 

\usepackage{lineno}
\let\doendproof\endproof
\renewcommand\endproof{~\hfill\qed\doendproof}

\renewcommand{\paragraph}[1]{\smallskip\noindent\textbf{#1}\xspace}
\graphicspath{{img/}}

\graphicspath{{figures/}}

\newcommand{\remove}[1]{}

\usepackage[nameinlink,capitalise]{cleveref}
\Crefname{observation}{Observation}{Observations}
\Crefname{algorithm}{Algorithm}{Algorithms}
\Crefname{section}{Section}{Sections}
\Crefname{observation}{Observation}{Observations}
\Crefname{lemma}{Lemma}{Lemmas}
\Crefname{claim}{Claim}{Claims}
\Crefname{figure}{Fig.}{Figs.}
\Crefname{figure}{Fig.}{Figs.}
\Crefname{enumi}{Condition}{Conditions}
\Crefname{property}{Property}{Properties}

\usepackage{color}
\definecolor{realblue}{rgb}{0,0,1}
\definecolor{blue}{rgb}{0.274,0.392,0.666}
\definecolor{darkerblue}{rgb}{0.094,0.455,0.804}
\definecolor{darkblue}{rgb}{0.063,0.306,0.545}
\definecolor{red}{rgb}{0.627,0.117,0.156}
\definecolor{green}{rgb}{0,0.588,0.509}
\definecolor{orange}{rgb}{0.903,0.739,0.382}
\definecolor{realred}{rgb}{1,0,0}

\newcommand{\darkblue}[1]{{{\textcolor{darkblue}{#1}\xspace}}}

\hypersetup{colorlinks,linkcolor={darkblue},citecolor={darkblue},urlcolor={darkblue}} 

\newcommand{\acks}{
	This research was supported by MIUR Proj.\ ``MODE'' n$^\circ$ 20157EFM5C, 
	by MIUR Proj.\ ``AHeAD'' n$^\circ$ 20174LF3T8, 
	by MIUR-DAAD JMP n$^\circ$ 34120,
	by H2020-MSCA-RISE Proj.\ ``CONNECT'' n$^\circ$ 734922, and
	by Roma Tre University Proj.\ ``GeoView''.\xspace
}

\newcommand{\etal}{{\em et~al.}\xspace}
\renewcommand{\emph}[1]{\darkblue{\em #1}}
\begin{document}
%
\title{Graph Stories in Small Area\thanks{\acks}}

%
\author{Manuel Borrazzo \and Giordano {Da Lozzo} \and Giuseppe {Di Battista} \and Fabrizio Frati \and Maurizio Patrignani}
\authorrunning{M. Borrazzo \etal}
\institute{Roma Tre University, Rome, Italy\\ 
\href{mailto:manuel.borrazzo@uniroma3.it,giordano.dalozzo@uniroma3.it,fabrizio.frati@uniroma3.it,maurizio.patrignani@uniroma3.it}{firstname.lastname@uniroma3.it}}
\maketitle             
\begin{abstract}
We study the problem of drawing a dynamic graph, where each vertex appears in the graph at a certain time and remains in the graph for a fixed amount of time, called the window size. 
This defines a graph story, i.e., a sequence of subgraphs, each induced
by the vertices that are in the graph at the same time.
The drawing of a graph story is a sequence of drawings of such subgraphs. 
To  support readability, we require that
each drawing is straight-line and planar and that each vertex maintains its placement in all the drawings.
Ideally, the area of the drawing of each subgraph should be a function of the window size, rather than a function of the size of the entire graph, which could be too large. 
We show that the graph stories of paths and trees can be drawn on a $2W \times 2W$ and on an $(8W+1) \times (8W+1)$ grid, respectively, where $W$ is the window size. These results are constructive and yield linear-time algorithms.
Further, we show that there exist graph stories of planar graphs whose subgraphs cannot be drawn within an area that is only a~function~of~$W$.
\end{abstract}

\section{Introduction}\label{se:intro}



We consider a graph that changes over time. Its vertices enter the graph one after the other and persist in the graph for a fixed amount of time, called the \emph{window size}. We call such a dynamic graph a \emph{graph~story}. More formally, let $V$ be the set of vertices of a graph $G$. Each vertex $v \in V$ is equipped with a label~$\tau(v)$, which specifies the time instant at which $v$ appears in the graph. The labeling $\tau: V \rightarrow \{1,2,\dots, |V|\}$ is a bijective function specifying a total ordering~for~$V$.
At any time $t$, the graph $G_t$ is the subgraph of $G$ induced by the set of vertices $\{v \in V: t-W < \tau(v)\leq t\}$. We denote a graph story by $\langle G, \tau, W \rangle$.

We are interested in devising an algorithm for visualizing graph stories. The input of the algorithm is an entire graph story and the output is what we call a drawing story. A \emph{drawing story} is a sequence of drawings of the graphs $G_t$. The typical graph drawing conventions can be applied to a drawing story. E.g., a drawing story is planar, straight-line, or on the grid if all its drawings are planar, straight-line, or on the grid, respectively. 

A trivial way for constructing a drawing story would be to first produce a drawing of $G$ and then to obtain a drawing of $G_t$, for each time $t$, by filtering out vertices and edges that do not belong to $G_t$. However, if the number of vertices of $G$ is much larger than $W$, this strategy might produce unnecessarily large drawings. Ideally, the area of the drawing of each graph $G_t$ should be a (polynomial) function of $W$, rather than a function of the size of the entire graph.

In this paper we show that the graph stories of paths and trees can be drawn on a $2W \times 2W$ and on an $(8W+1) \times (8W+1)$ grid, respectively, so that all the drawings of the story are straight-line and planar, and so that vertices do not change their position during the drawing story. Further, we show that there exist graph stories of planar graphs that cannot be drawn within an area that is only a function of $W$, if planarity is required and vertices are not allowed to change their position during the drawing story.


The visualization of dynamic graphs is a classic research topic in graph drawing. In what follows we compare our model and results with the literature. 
We can broadly classify the different approaches in terms of the following features~\cite{DBLP:conf/vissym/0001B0W14}. 
\begin{inparaenum}[(i)]
	\item The objects that appear and disappear over time can be vertices or edges.
	\item The lifetime of the objects may be fixed or variable.
	\item The story may be entirely known in advance ({\em off-line model}) or not ({\em on-line model}).
\end{inparaenum}
In this paper, the considered objects are vertices, the lifetime is fixed and the model is off-line. 

A considerable amount of the literature on the theoretical aspects of dynamic graphs focuses on trees. In~\cite{DBLP:journals/ipl/BinucciBBDGPPSZ12}, the objects are edges, the lifetime $W$ is fixed, and the model is on-line; an algorithm is shown for drawing a tree in $O(W^3)$ area, under the assumption that the edges arrive in the order of a Eulerian tour of the tree. In~\cite{DBLP:conf/gd/DemetrescuBFLPP99}, the objects are vertices, the lifetime $W$ is fixed, the model is off-line (namely, the sequence of vertices is known in advance, up to a certain threshold $k$), and the vertices can move. A bound in terms of $W$ and $k$ is given on the total amount of movement of the vertices. In~\cite{DBLP:conf/gd/SkambathT16}, each subgraph of the story is given (each subgraph is a tree, whereas the entire graph may be arbitrary), each object can have an arbitrary lifetime, the model is off-line, and the vertices can move. Aesthetic criteria as in the classical Reingold-Tilford algorithm~\cite{DBLP:journals/tse/ReingoldT81} are pursued.


Other contributions consider more general types of graphs. In~\cite{DBLP:conf/gd/GoodrichP13a}, the objects are edges, which enter the drawing and never leave it, the model is on-line, the considered graphs are outerplanar, and the vertices are allowed to move by a polylogarithmic distance. In~\cite{DBLP:conf/ciac/LozzoR15}, the objects are edges, the lifetime is fixed, and the model is off-line; NP-completeness is shown for the problem of computing planar topological drawings of the graphs in the story; other results for the topological setting are presented in~\cite{DBLP:journals/algorithmica/AngeliniB19,DBLP:conf/gd/Schaefer14}.  

Further related results appear in~\cite{DBLP:journals/siamcomp/CohenBTT95,DBLP:journals/siamcomp/BattistaT96,DBLP:journals/iandc/BattistaTV01,DBLP:reference/algo/Italiano16e,DBLP:conf/stoc/Poutre94,DBLP:conf/gd/RextinH08}; in particular, geometric simultaneous embeddings~\cite{DBLP:reference/crc/BlasiusKR13,DBLP:journals/comgeo/BrassCDEEIKLM07} are closely related to the setting we consider in this paper. Contributions focused on the information-visualization aspects of dynamic graphs are surveyed in~\cite{DBLP:conf/vissym/0001B0W14}; further, a survey on temporal graph problems appears in~\cite{DBLP:journals/im/Michail16}.




\section{Preliminaries}
In this section, we present definitions and preliminaries.	

\paragraph{Graphs and drawings.}  We denote the set of vertices and edges of a graph $G$ by $V(G)$ and $E(G)$, respectively.  A \emph{drawing} of $G$ maps each vertex in~$V(G)$ to a distinct point in the plane and each edge in $E$ to Jordan arc between its end-points. A drawing is \emph{straight-line} if each arc is a straight-line segment, it is \emph{planar} if no two arcs intersect, except at a common endpoint, and it is \emph{on the grid} (or, a \emph{grid drawing}) if each vertex is mapped to a point with integer coordinates.


\paragraph{Rooted ordered forests and their drawings.}
A \emph{rooted tree} $T$ is a tree with one distinguished vertex, called \emph{root} and denoted by $r(T)$. For any vertex $v \in V(T)$ with $v\neq r(T)$, consider the unique path $\mathcal P_v$ between $v$ and $r(T)$ in~$T$; the \emph{ancestors} of $v$ are the vertices of $\mathcal P_v$, and the \emph{parent} $p(v)$ of $v$ is the~neighbor~of~$v$~in~$\mathcal P_v$. For any two vertices $u,v\in V(T)$, the \emph{lowest common ancestor of} $u$ \emph{and} $v$ is the ancestor of $u$ and $v$ whose graph-theoretic distance from $r(T)$ is maximum. For any vertex $v\in V(T)$ with $v\neq r(T)$, the \emph{children} of $v$ are the neighbors of $v$ different from $p(v)$; the \emph{children} of $r(T)$ are all its neighbors. We denote by $T(u)$ the subtree of $T$ rooted at a node $u$.

A \emph{rooted ordered tree}  $T$ is a rooted tree such that, for each vertex $v\in V(T)$, a \emph{left-to-right} (linear) order $u_1,\dots,u_k$ of the children of $v$ is specified. A sequence $\mathcal F = T_1,T_2,\dots,T_k$ of rooted ordered trees is a \emph{rooted ordered forest}. 

A \emph{strictly-upward} drawing of a rooted tree $T$ is such that each edge is represented by a curve monotonically increasing in the \mbox{$y$-direction} from a vertex to its parent.
A \emph{strictly-upward} drawing $\Gamma$ of a rooted forest $\mathcal F$ is such that the drawing of each tree $T_i \in \mathcal F$ in $\Gamma$ is strictly-upward.
A strictly-upward drawing of an ordered tree $T$ is \emph{order-preserving} if, for each vertex $v\in V(T)$, the left-to-right order of the edges from $v$ to its children in the drawing is the same as the order associated with $v$ in $T$. A strictly-upward drawing of a rooted ordered forest $\mathcal F$ is \emph{order-preserving} if the drawing of~each tree $T_i \in \mathcal F$ is order-preserving. 
The definitions of (\emph{order-preserving}) \emph{strictly-leftward}, \emph{strictly-downward}, and \emph{strictly-rightward} drawings of (ordered) rooted trees and forests are similar.

\paragraph{Geometry}.
Given two points $p_1$ and $p_2$ in $\mathbb{R}^2$, we denote by $[p_1,p_2]$ the closed straight-line segment connecting $p_1$ and $p_2$.
Given two closed intervals $[a,b]$ and $[c,d]$, where $a$, $b$, $c$, and $d$ are integers with $a<b$ and $c < d$, the \emph{(integer) grid $[a,b] \times [c,d]$} is the set of points given by the Cartesian product $([a,b] \cap \mathbb{Z}) \times ([c,d] \cap \mathbb{Z})$.
The \emph{width} and \emph{height} of $[a,b] \times [c,d]$ are $|[a,b] \cap \mathbb{Z}|=b-a+1$ and $|[c,d] \cap \mathbb{Z}|=d-c+1$, respectively.
A grid drawing $\Gamma$ of a graph \emph{lies on a} $W\times H$ grid if $\Gamma$ is enclosed by the bounding box of some grid of width $W$ and height $H$, and \emph{lies on the} grid $[a,b] \times [c,d]$ if $\Gamma$ is enclosed by the bounding box of the grid  $[a,b] \times [c,d]$.

\paragraph{Graph stories.}
A graph story $\langle G, \tau, W \rangle$ is naturally \emph{associated with} a sequence $\langle G_1, G_2, \dots, G_{n+W-1} \rangle$; for any $t \in \{1,\dots,n+W-1\}$, the graph $G_t$ is the subgraph of $G$ induced by the set of vertices $\{v \in V: t-W < \tau(v)\leq t\}$. Clearly, $|V(G_t)| \leq  W$. Note that $G_1,G_2,\dots,G_{W-1}$ are subgraphs of $G_W$ and $G_{n+1},G_{n+2},\dots,G_{n+W-1}$ are subgraphs of $G_n$, while each of $G_t$ and $G_{t+1}$ has a vertex that the other graph does not have, for $t=W,\dots,n-1$. A graph story $\langle G, \tau, W\rangle$ in which $G$ is a planar graph, a path, or a tree is a \emph{planar graph story}, a \emph{path story}, or a \emph{tree story}, respectively.

A \emph{drawing story} for $\langle G, \tau, W \rangle$ is a sequence $\langle \Gamma_1, \Gamma_2, \dots, \Gamma_{n+W-1} \rangle$ of drawings such that, for every $t=1, \dots, n+W-1$:
\begin{enumerate}[(i)]
	\item $\Gamma_t$ is a drawing of $G_t$,
	\item a vertex $v$ is drawn at the same position in all the drawings $\Gamma_t$ such that $v \in V(G_t)$, and
	\item an edge $(u,v)$ is represented by the same curve in all the drawings $\Gamma_t$ such that $(u,v) \in E(G_t)$.
\end{enumerate}
The above definition implies that the drawings $\Gamma_1,\Gamma_2,\dots,\Gamma_{W-1}$ are the restrictions of $\Gamma_W$ to the vertices and edges of $G_1,G_2,\dots,G_{W-1}$, respectively, and that the drawings $\Gamma_{n+1},\Gamma_{n+2},\dots,\Gamma_{n+W-1}$ are the restrictions of $\Gamma_n$ to the vertices and edges of $G_{n+1},G_{n+2},\dots,G_{n+W-1}$, respectively. Hence, an algorithm that constructs a drawing story only has to specify the drawings $\Gamma_{W},\Gamma_{W+1},\dots,\Gamma_{n}$. 

In the remainder, we only consider drawing stories $\Gamma$ that are planar, straight-line, and on the grid. Storing each drawing in $\Gamma$ explicitly would require $\Omega(n \cdot W)$ space in total. However, since each vertex maintains the same position in all the drawings where it appears, since edges are straight-line segments, and since any two consecutive graphs in a graph story only differ by $O(1)$ vertices, we can encode $\Gamma$ in total $O(n)$ space.

Let $\Gamma$ be a straight-line drawing story of a graph story $\langle G, \tau, W \rangle$ and let $G'$ be a subgraph of $G$ (possibly $G'=G$). The \emph{drawing of $G'$ induced by $\Gamma$} is the straight-line drawing of $G'$ in which each vertex has the same position as in every drawing $\Gamma_t\in \Gamma$ where it appears. Note that the drawing of $G'$ induced by $\Gamma$ might have crossings even if $\Gamma$ is planar. For a subset $B \subseteq V(G)$, let $G[B]$ be the subgraph of $G$ induced by the vertices in $B$ and let $\Gamma[B]$ be the straight-line drawing of $G[B]$ induced by $\Gamma$.

Given a graph story $\langle G, \tau, W \rangle$, we will often consider a partition of $V$ into \emph{buckets} $B_1, \dots, B_h$, where $h=\lceil \frac{n}{W}\rceil$. For $i=1,\dots,h$, the bucket $B_i$ is the set of vertices $v$ such that $(i-1)W+1 \leq \tau(v) \leq \min \{i\cdot W,n\}$. Note that all buckets have $W$ vertices, except, possibly, for $B_h$.  We have the following useful property.

\begin{property}\label{prop:supergraph}
	For any $t=1,2,\dots,n+W-1$, let $i=\lceil \frac{t}{W}\rceil$. Then $G_t$ is a subgraph of $G[B_{i-1} \cup B_i]$.
\end{property}

\begin{proof}
We have $B_{i-1} \cup B_i=\{v\in V(G): (i-2)W+1 \leq \tau(v) \leq \min (i\cdot W,n)\}$, which is a superset of $V(G_t)=\{v\in V(G): t-W < \tau(v) \leq \min (t,n)\}$, given that $(i-2)W+1 \leq t-W$ and that $i\cdot W \geq t$. 
\end{proof}  

\section{Planar Graph Stories}\label{se:planar}

In this section we prove a lower bound on the size of any drawing story of a planar graph story. As common in the literature about small-area graph drawings, such a lower bound is achieved by means of the nested triangles graph. 


Let $n\equiv 0 \mod 3$. An $n$-vertex \emph{nested triangles graph} $G$ contains the vertices and the edges of the $3$-cycle $C_{i} = (v_{i-2},v_{i-1},v_i)$, for $i=3,6,\dots,n$, plus the edges $(v_i,v_{i+3})$, for $i=1,2,\dots,n-3$. The nested triangles graphs with $n\geq 6$ vertices are $3$-connected and thus they have a unique combinatorial embedding (up to a flip)~\cite{Whitney33}. We have the following.

\begin{theorem}\label{th:lowerbound}
	Let $\langle G, \tau, 9\rangle$ be a planar graph story such that $G$ is an $n$-vertex nested triangles graph and $\tau(v_i)=i$. Then any drawing story of $\langle G, \tau, 9\rangle$ lies on a $\Omega(n) \times \Omega(n)$ grid.
\end{theorem}

\begin{proof}
	In order to prove the statement we show that, for any drawing story $\Gamma$ of $\langle G, \tau, 9\rangle$, the straight-line drawing of $G$ induced by $\Gamma$ is planar. Then the statement follows from well-known lower bounds in the literature~\cite{DBLP:journals/combinatorica/FraysseixPP90,Dolev,DBLP:conf/gd/FratiP07}.
	
	Let $\Gamma = \Gamma_1,\Gamma_2,\dots,\Gamma_{n+8}$. Note that, for any $m=9,12,\dots,n$, the graph $G_m$ contains the cycles $C_{m-6}=(v_{m-8},v_{m-7},v_{m-6})$, $C_{m-3}=(v_{m-5},v_{m-4},v_{m-3})$, and $C_{m}=(v_{m-2},v_{m-1},v_{m})$. For any $m=9,12,\dots,n$, let $H_m$ and $I_m$ be the subgraphs of $G$ induced by $v_1,v_2,\dots,v_{m}$ and by $v_{m-5},v_{m-4},\dots,v_{m}$, respectively, and let $\Gamma^H_m$ and $\Gamma^I_m$ be the drawings of $H_m$ and $I_m$ induced by $\Gamma$, respectively.
	
	Since $G=H_n$, in order to show that the drawing of $G$ induced by $\Gamma$ is planar, it suffices to show that $\Gamma^H_m$ is planar, for each $m=9,12,\dots,n$. We prove this statement by induction on $m$. The induction also proves that the cycle $C_m$ bounds a face of $\Gamma^H_m$ and of $\Gamma^I_m$. 
	
	
	Suppose first that $m=9$. Then $H_9=G_9$ and the drawing $\Gamma^H_9=\Gamma_9$ of $G_9$ is planar since $\Gamma$ is a planar straight-line drawing of $\langle G, \tau, 9\rangle$ and $\Gamma_9 \in \Gamma$. The $3$-connectivity of $H_9$ and $I_9$ implies that $C_9$ bounds a face of $\Gamma^H_9$ and $\Gamma^I_9$. 
	
	
	Suppose now that $m>9$. We show that $\Gamma^H_m$ is planar. By induction, $\Gamma^H_{m-3}$ is planar. Thus, we only need to prove that no crossing is introduced by placing the vertices $v_{m-2}$, $v_{m-1}$, and $v_m$, which belong to $H_m$ and not to $H_{m-3}$, in $\Gamma^H_{m-3}$ as they are placed in $\Gamma_m$ and by drawing their incident edges as straight-line segments. First, by induction, the cycle $C_{m-3}=(v_{m-5},v_{m-4},v_{m-3})$ bounds a face of $\Gamma^H_{m-3}$ and of $\Gamma^I_{m-3}$. Second, the vertices $v_{m-2}$, $v_{m-1}$, and $v_{m}$, as well as their incident edges, lie inside such a face in $\Gamma_m$. Namely, $v_{m-2}$, $v_{m-1}$, and $v_{m}$ lie inside the same face of $\Gamma^I_{m-3}$ in $\Gamma_m$, as otherwise the cycle $C_m$ would cross edges of $I_{m-3}$ in $\Gamma_m$; further, the face of $\Gamma^I_{m-3}$ in which  $v_{m-2}$, $v_{m-1}$, and $v_{m}$ lie in $\Gamma_m$ is incident to all of $v_{m-5}$, $v_{m-4}$, and $v_{m-3}$, as otherwise the edges $(v_{m-5},v_{m-2})$, $(v_{m-4},v_{m-1})$, and $(v_{m-3},v_{m})$ would cross edges of $I_{m-3}$ in $\Gamma_m$; however, no face of $\Gamma^I_{m-3}$ other than the one delimited by $C_{m-3}$ is incident to all of $v_{m-5}$, $v_{m-4}$, and $v_{m-3}$. This proves the planarity of $\Gamma^H_{m}$. The induction and the proof of the theorem are completed by observing that $C_{m}$ bounds a face of $\Gamma^H_m$ and $\Gamma^I_m$.
\end{proof}

\section{Path Stories}\label{se:paths}

Let $\langle G, \tau, W\rangle$ be a path story, where $G=(v_1, v_2, \dots, v_n)$.  Note that the ordering of $V(G)$ given by the path is, in general, different from the ordering given by $\tau$.  


The \emph{$x$-buckets} of $G$ are the sets $B^x_i$, with $i = 1, \dots, \lceil\frac{h+1}{2}\rceil$, such that $B^x_1=B_1$, and $B^x_i = B_{2i-2} \cup B_{2i-1}$, for $i = 2, \dots, \lceil\frac{h+1}{2}\rceil$. Note that each $x$-bucket has $2W$ vertices, except for $B_1$ and, possibly, the last $x$-bucket; see \cref{fig:xy-buckets}.
The \emph{$y$-buckets} of $G$ are the sets $B^y_j$, with $j = 1, \dots, \lceil\frac{h}{2}\rceil$, such that  $B^y_j = B_{2j-1} \cup B_{2j}$. Note that each $y$-bucket has $2W$ vertices, except possibly for the last $y$-bucket; see \cref{fig:xy-buckets}.
Also, each vertex belongs to exactly one $x$-bucket and to exactly one $y$-bucket.

\begin{figure}[tb!]
\centering
\includegraphics[width=\textwidth]{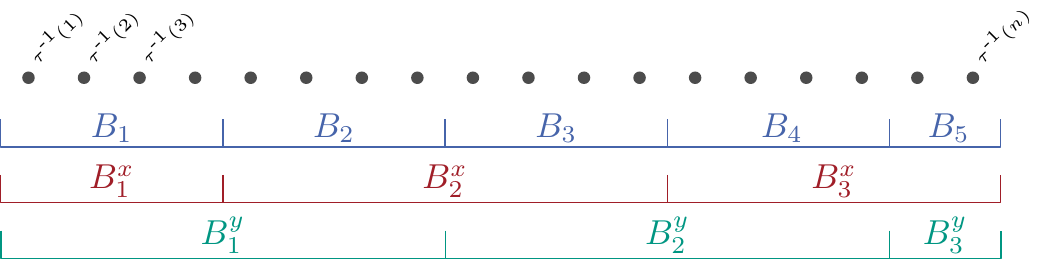}
\caption{Examples of buckets, $x$-buckets, and $y$-buckets when $W=4$.}\label{fig:xy-buckets}
\end{figure}

The following theorem is the contribution of this section; its proof is similar in spirit to the proof that any two paths admit a simultaneous geometric embedding~\cite{DBLP:journals/comgeo/BrassCDEEIKLM07}. 

\begin{theorem}\label{th:paths}
For any path story $\langle G, \tau, W \rangle$, it is possible to compute in $O(n)$ time a drawing story that is planar, straight-line, and lies on a~$2W \times 2W$~grid.
\end{theorem}

\begin{proof}
Let $G$ be the path $(v_1, v_2, \dots, v_n)$ and let $h=\lceil \frac{n}{W}\rceil$ be the number of buckets of $V(G)$. We now order the vertices in each $x$-bucket and in each $y$-bucket; this is done according to the ordering in which the vertices appear~in~the~path~$G$. Formally, for $i=1,2,\dots,\lceil \frac{h+1}{2}\rceil$, let $x_i: B^x_i \rightarrow \{1,\dots,|B^x_i|\}$ be a bijective function such that, for any $v_a,v_b \in B^x_i$, we have $x_i(v_a)<x_i(v_b)$ if and only if $a<b$.
Similarly, for $i=1,2,\dots,\lceil \frac{h}{2}\rceil$, let $y_i: B^y_i \rightarrow \{1,\dots,|B^y_i|\}$ be a bijective function such that, for any $v_a,v_b \in B^y_i$, we have $y_i(v_a)<y_i(v_b)$ if and only if $a<b$.
We assign the coordinates to the vertices of~$G$ in $\Gamma$ as follows. 
For any vertex $v$ of~$G$, let $B^x_i$ and $B^y_j$ be the $x$-bucket and the $y$-bucket containing $v$, respectively. We place $v$ at the point $(x_i(v),y_j(v))$ in all the drawings $\Gamma_t \in \Gamma$ such that $v$ belongs to~$G_t$. Also, we draw each edge as a straight-line segment. This concludes the construction~of~$\Gamma$.

We now prove that $\Gamma$ satisfies the properties in the statement.

Since $x$-buckets and $y$-buckets have size at most~$2W$, we have that each vertex is assigned integer $x$- and $y$-coordinates in the interval $[1,2W]$. Thus, each drawing $\Gamma_t 	\in \Gamma$ lies on the $[1,2W] \times [1,2W]$ grid.

We show that each drawing $\Gamma_t \in \Gamma$ is planar. 
\begin{property}\label{cl:monotonicity}
For $i=1,\dots, \lceil \frac{h+1}{2}\rceil$, the straight-line drawing $\Gamma[B^x_i]$ of $G[B^x_i]$ is planar.
For  $j=1,\dots,\lceil \frac{h}{2}\rceil$, \mbox{the straight-line drawing $\Gamma[B^y_j]$ of $G[B^y_j]$ is planar.}
\end{property}

\begin{proof}
Consider any $x$-bucket $B^x_i$. By construction, the drawing $\Gamma[B^x_i]$ of $G[B^x_i]$ is \emph{$x$-monotone}, that is, the left-to-right order in $\Gamma[B^x_i]$ of the vertices in $B^x_i$  coincides with the order of such vertices in $G$. Hence, $\Gamma[B^x_i]$ is planar. Analogously, the drawing $\Gamma[B^y_j]$ of $G[B^y_j]$ is $y$-monotone, hence it is planar.
\end{proof}


We show that, for $t=1,\dots,n+W-1$, the straight-line drawing $\Gamma_t \in \Gamma$ of the graph $G_t$ is planar. By \cref{prop:supergraph}, there exists an $x$-bucket $B^x_i$ or a $y$-bucket $B^y_j$ that contains $V(G_t)$. The drawing $\Gamma_t$ coincides with the drawing $\Gamma[B^x_i]$ or with the drawing $\Gamma[B^y_j]$, respectively, restricted to the vertices in $V(G_t)$. By \cref{cl:monotonicity}, we have that $\Gamma[B^x_i]$ and $\Gamma[B^y_j]$ are planar, and hence so is $\Gamma_t$.

Finally, the time needed to compute $\Gamma$ coincides with the time needed to compute the functions $x_i$  and $y_j$, for each $i \in \{1,2,\dots, \lceil \frac{h+1}2\rceil\}$ and $j \in \{1,2,\dots, \lceil \frac{h}2\rceil \}$. 
This can be done in total $O(n)$ time as follows. For $i=1,\dots,n$, traverse the path $G$ from $v_1$ to $v_n$. When a vertex $v_k$ is considered, the buckets $B^x_i$ and $B^y_j$ where $v_k$ should be inserted are determined in $O(1)$ time; then $v_k$ is inserted in each of these buckets as the currently last vertex. This process provides each $x$-bucket $B^x_i$ and each $y$-bucket $B^y_j$ with the desired orderings $x_i$ and $y_j$, respectively.
\end{proof}


\newcommand{\sewedge}{\raisebox{-2pt}{\includegraphics[page=1,scale=0.45]{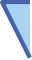}}\hspace{1pt}}
\newcommand{\swwedge}{\raisebox{-2pt}{\includegraphics[page=2,scale=0.45]{wedges.pdf}}\hspace{1pt}}
\newcommand{\newedge}{\raisebox{-2pt}{\includegraphics[page=3,scale=0.45]{wedges.pdf}}\hspace{1pt}}
\newcommand{\nwwedge}{\raisebox{-2pt}{\includegraphics[page=4,scale=0.45]{wedges.pdf}}\hspace{1pt}}
\newcommand{\sedrawing}{\raisebox{-1pt}{\includegraphics[page=5,scale=0.45]{wedges.pdf}}\hspace{1pt}}
\newcommand{\swdrawing}{\raisebox{-1pt}{\includegraphics[page=6,scale=0.45]{wedges.pdf}}\hspace{1pt}}
\newcommand{\nwdrawing}{\raisebox{-1pt}{\includegraphics[page=7,scale=0.45]{wedges.pdf}}\hspace{1pt}}
\newcommand{\nedrawing}{\raisebox{-1pt}{\includegraphics[page=8,scale=0.45]{wedges.pdf}}\hspace{1pt}}
\newcommand{\oursearrow}{\raisebox{-2pt}{\includegraphics[page=9,scale=0.35]{wedges.pdf}}\hspace{1pt}}
\newcommand{\ourswarrow}{\raisebox{-2pt}{\includegraphics[page=10,scale=0.35]{wedges.pdf}}\hspace{1pt}}
\newcommand{\ournearrow}{\raisebox{-2pt}{\includegraphics[page=11,scale=0.35]{wedges.pdf}}\hspace{1pt}}
\newcommand{\ournwarrow}{\raisebox{-2pt}{\includegraphics[page=12,scale=0.35]{wedges.pdf}}\hspace{1pt}}

\section{Tree Stories}\label{se:trees}

In this section we show how to draw a tree story $\langle T, \tau, W \rangle$. Our algorithm partitions $V(T)$ into buckets $B_1,\dots,B_h$ and then partitions the subtrees of $T$ induced by each bucket $B_i$ into two rooted ordered forests. For odd values of $i$, the forests corresponding to $B_i$ are drawn ``close'' to the $y$-axis, while for even values of $i$, the forests corresponding to $B_i$ are drawn ``close'' to the $x$-axis. The drawings of these forests need to satisfy strong visibility properties, as edges of $T$ might connect vertices in a bucket $B_i$ with the roots of the forests corresponding to the bucket $B_{i+1}$, and vice versa. We now introduce a drawing standard for (static) rooted ordered forests that guarantees these visibility properties.

For a vertex $v$ in a drawing $\Gamma$, denote by $\ell_{\oursearrow}(v)$ the half-line originating at $v$ with slope $-2$. Also, consider a horizontal half-line originating at $v$ and directed rightward; rotate such a line in clockwise direction around $v$ until it overlaps with~$\ell_{\oursearrow}(v)$; this rotation spans a closed wedge centered at $v$, which we call the \emph{$\sewedge$-wedge} of $v$ and denote by~$S_{\sewedge}(v)$. We have the following definition.


\begin{definition} \label{def:quadrant-drawings}
Let $\mathcal F =T_1,T_2,\dots,T_k$ be a rooted ordered forest, with a total of $m \leq W$ vertices. A~$\sedrawing$-\emph{drawing}~$\Gamma$ of $\mathcal F$ is a planar straight-line strictly-upward strictly-leftward order-preserving grid drawing of~$\mathcal F$ such that:
\begin{enumerate}[(i)]
\item \label{prop:grid} $\Gamma$ lies on the $[0,m-1] \times [4W-2m+2,4W]$ grid;
\item \label{prop:bottom-to-top-trees} the roots $r(T_1),r(T_2),\dots,r(T_k)$  lie along the segment $\big[(0,2W+2),(0,4W)\big]$, in this order from bottom to top, and $r(T_k)$ lies on $(0,4W)$;
\item \label{prop:bottom-to-top-sub-trees} the vertices of $T_i$ have $y$-coordinates strictly smaller than the vertices of $T_{i+1}$, for $i=1,\dots,k-1$;
\item \label{prop:bottom-to-top-internal} for each tree $T_i$ and for each vertex $v$ of $T_i$, let $u_1,u_2,\dots,u_{\ell}$ be the children of $v$ in left-to-right order; then the vertices of $T_i(u_j)$ have $y$-coordinates strictly smaller than the vertices of $T_i(u_{j+1})$, for $j=1,\dots,\ell-1$; and
 \item \label{prop:wedge-free-internally} for each vertex $v$ of $\mathcal F$, the wedge $S_{\sewedge}(v)$ does not intersect $\Gamma$ other than along $\ell_{\oursearrow}(v)$. 
\end{enumerate}
\end{definition}

We are going to use the following properties.

\begin{property}\label{prop:drawing-inside-triangle}
For $i=1,\dots,k$, the drawings of $T_1,T_2,\dots,T_i$ in a $\sedrawing$-drawing $\Gamma$ of $\mathcal F$ lie inside or on the boundary of the triangle delimited by the $y$-axis, by $\ell_{\oursearrow}(r(T_i))$, and by the horizontal line $y:=2W+2$.
\end{property}

\begin{proof}
That $\Gamma$ lies to the right or along the $y$-axis follows by \cref{prop:bottom-to-top-trees} of \cref{def:quadrant-drawings} and by the fact that  $\Gamma$ is strictly-leftward. That $\Gamma$ lies above or along the horizontal line $y:=2W+2$ follows by \cref{prop:grid} of \cref{def:quadrant-drawings} and by $m\leq W$. Finally, we prove that the drawings of $T_1,T_2,\dots,T_i$ in $\Gamma$ lie below or along $\ell_{\oursearrow}(r(T_i))$. That the drawing of $T_i$ in $\Gamma$ lies below or along $\ell_{\oursearrow}(r(T_i))$ follows by \cref{prop:wedge-free-internally} of \cref{def:quadrant-drawings} and by the fact that $\Gamma$ is strictly-upward. Further, for $j=1,\dots,i-1$, \cref{prop:bottom-to-top-trees} of \cref{def:quadrant-drawings} implies that $\ell_{\oursearrow}(r(T_j))$ lies below $\ell_{\oursearrow}(r(T_i))$; since the drawing of $T_j$ in $\Gamma$ lies below or along $\ell_{\oursearrow}(r(T_j))$, it follows that the drawing of $T_j$ in $\Gamma$ lies below $\ell_{\oursearrow}(r(T_i))$.
\end{proof}

\begin{property}\label{prop:containment}
For each vertex $v$ in a $\sedrawing$-drawing $\Gamma$ of $\mathcal F$, the wedge $S_{\sewedge}(v)$ contains the segment $\big[(2W+2,0),(4W,0)\big]$ in its interior.
\end{property}

\begin{proof}
The proof consists of two parts. 

First, the wedge $S_{\sewedge}(r(T_k))$ contains the segment $\big[(2W+2,0),(4W,0)\big]$ in its interior. Namely, by~\cref{prop:bottom-to-top-trees} of \cref{def:quadrant-drawings}, the vertex $r(T_k)$ is placed at $(0,4W)$, hence the half-line $\ell_{\oursearrow}(r(T_k))$ intersects the $x$-axis at the point $(2W,0)$.

Second, by \cref{prop:drawing-inside-triangle}, every vertex $v$ of $\mathcal F$ lies inside or on the boundary of the triangle delimited by the $y$-axis, by $\ell_{\oursearrow}(r(T_k))$, and by the horizontal line $y:=2W+2$. Hence $\ell_{\oursearrow}(v)$ either overlaps or is below $\ell_{\oursearrow}(r(T_k))$, which implies that $\ell_{\oursearrow}(v)$ intersects the $x$-axis at a point $(r,0)$ with $r\leq 2W$.
\end{proof}

We can similarly define $\swdrawing$-, $\nwdrawing$-, and $\nedrawing$-\emph{drawings}; in particular, a drawing of $\mathcal F$ is a $\swdrawing$-, $\nwdrawing$-, or $\nedrawing$-drawing if and only if it can be obtained from a $\sedrawing$-drawing by a clockwise rotation around the origin of the Cartesian axes by $90^\circ$, $180^\circ$, or $270^\circ$, respectively. A property similar to \cref{prop:containment} can be stated for such drawings.

We now present an algorithm, called $\sedrawing$-{\sc Drawer}, that constructs a $\sedrawing$-drawing $\Gamma$ of $\mathcal F =T_1,\dots,T_k$. The algorithm $\sedrawing$-{\sc Drawer} constructs $\Gamma$ by induction, primarily, on $k$ and, secondarily, on the number of vertices $m$ of $\mathcal F$. We only describe how to place the vertices of $\mathcal F$. Since $\Gamma$ is a straight-line drawing, this also determines the representation of the edges~of~$\mathcal F$.

\begin{figure}[t!]
\centering
\includegraphics[page=2,width=0.7\textwidth]{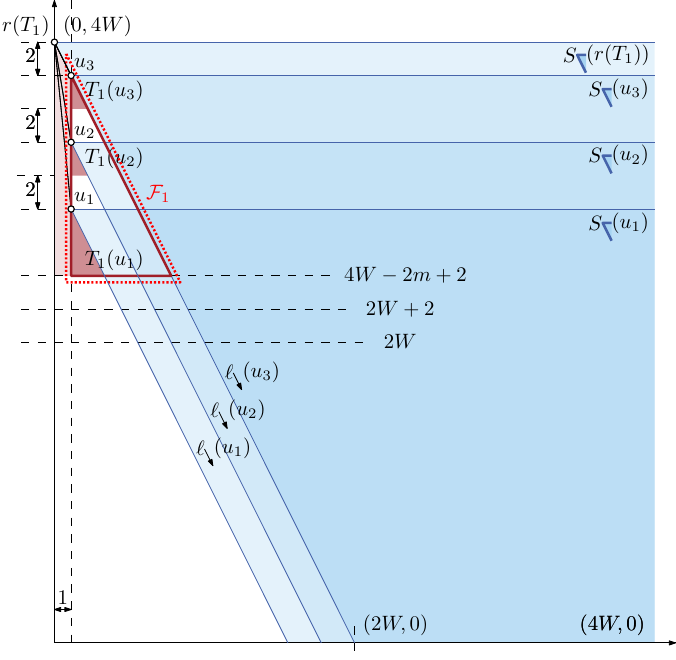}
\caption{Construction of a $\protect\sedrawing$-drawing: the first inductive case, in which $k=1$ and $m>1$.}
\label{fi:induction-first}
\end{figure}	

\begin{figure}[t!]
\centering
\includegraphics[page=3,width=0.7\textwidth]{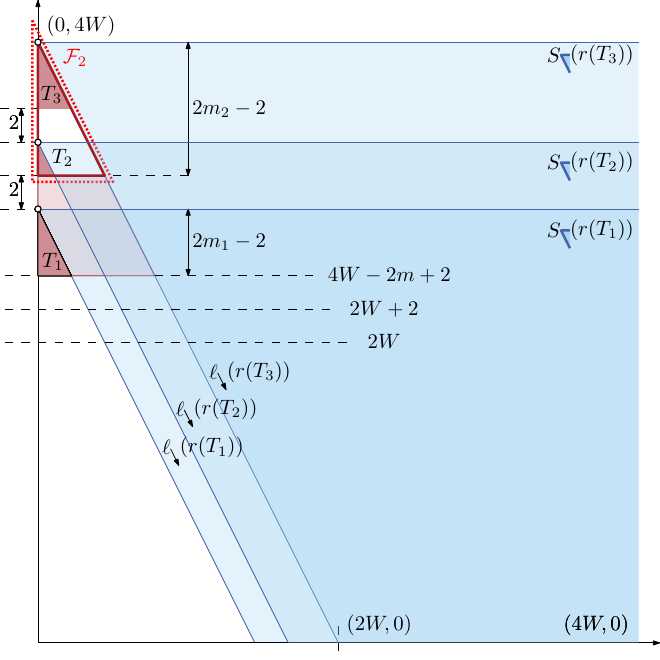}
\caption{Construction of a $\protect\sedrawing$-drawing: the second inductive case, in which $k>1$.}
\label{fi:induction-second}
\end{figure}

The base case of the algorithm $\sedrawing$-{\sc Drawer} happens when $m=1$ (and thus $k=1$); then we obtain $\Gamma$ by placing $r(T_1)$ at $(0,4W)$.

In the first inductive case, we have $k=1$ and $m>1$; see \cref{fi:induction-first}. Let $\mathcal F_1$ be the  rooted ordered forest $T_1(u_1),T_1(u_2),\dots,T_1(u_{\ell_1})$, where $u_1,u_2,\dots,u_{\ell_1}$ are the children of $r(T_1)$ in left-to-right order. Inductively construct a  $\sedrawing$-drawing $\Gamma_1$ of $\mathcal F_1$. 
We obtain $\Gamma$ by placing $r(T_1)$ at $(0,4W)$ and by translating $\Gamma_1$ one unit to the right and two units down.

In the second inductive case, we have $k>1$; see \cref{fi:induction-second}. We inductively construct a $\sedrawing$-drawing $\Gamma_1$ of $T_1$ and a $\sedrawing$-drawing $\Gamma_2$ of the rooted ordered forest $\mathcal F_2 = T_2,T_3,\dots,T_{k}$. Then, we obtain $\Gamma$ from $\Gamma_1$ and $\Gamma_2$ by translating $\Gamma_1$ down so that $r(T_1)$ lies two units below the lowest vertex in $\Gamma_2$.
We now prove the following.

\begin{lemma}\label{le:quadrant-drawings}
The algorithm $\sedrawing$-{\sc Drawer} constructs a $\sedrawing$-drawing of $\mathcal F$~in~$O(m)$~time.
\end{lemma}


\begin{proof}
The algorithm $\sedrawing$-{\sc Drawer} clearly runs in $O(m)$ time. We now prove that the drawing $\Gamma$ of $\mathcal F$ it constructs is a $\sedrawing$-drawing. This is trivial in the base case.



In the first inductive case, we have $k=1$ and $m>1$. Recall that $\Gamma$ has been obtained by placing $r(T_1)$ at $(0,4W)$ and by translating the inductively constructed $\sedrawing$-drawing $\Gamma_1$ of $\mathcal F_1=T_1(u_1),T_1(u_2),\dots,T_1(u_{\ell_1})$ one unit to the right and two units below. Let $\Gamma'_1$ be the drawing of $\mathcal F_1$ in $\Gamma$.

First, we prove that $\Gamma$ is planar. By construction, $\Gamma'_1$ is a translation of $\Gamma_1$, which is planar by induction. Further, by construction, the edges incident to $r(T_1)$ lie to the left of the vertical line $x:=1$, except at the vertices $u_1,u_2,\dots,u_{\ell_1}$; also, $\Gamma'_1$ lies to the right of the vertical line $x:=1$, except at the vertices $u_1,u_2,\dots,u_{\ell_1}$. Hence, no edge incident to $r(T_1)$ intersects any edge of $\mathcal F_1$ in $\Gamma$ other than at a common end-point. Finally, no two edges incident to $r(T_1)$ overlap, since their end-points different from $r(T_1)$ lie at distinct points on the vertical line $x:=1$, by \cref{prop:bottom-to-top-trees} of \cref{def:quadrant-drawings} for $\Gamma_1$.

Second, we prove that $\Gamma$ is strictly-upward, strictly-leftward, and order-preserving. By construction, $\Gamma'_1$ is a translation of $\Gamma_1$, which is strictly-upward, strictly-leftward, and order-preserving, by induction. Further, by construction, $r(T)$ lies above and to the left of $\Gamma'_1$, and hence it lies above and to the left of its children $u_1,u_2,\dots,u_{\ell_1}$; it follows that $\Gamma$ is strictly-upward and strictly-leftward. Also, since $u_1,u_2,\dots,u_{\ell_1}$ appear in this bottom-to-top order along the vertical line $x:=1$, by \cref{prop:bottom-to-top-trees} of \cref{def:quadrant-drawings} for $\Gamma_1$, it follows that $\Gamma$ is order-preserving.

Third, that $\Gamma$ is a grid drawing easily follows from the fact that $\Gamma'_1$ is a translation by an integer vector of $\Gamma_1$, which is a grid drawing by induction, and from the fact that $r(T_1)$ is placed at $(0,4W)$, by construction.

We now prove that $\Gamma$ satisfies~\cref{prop:grid} of \cref{def:quadrant-drawings}. Note that $\mathcal F_1$ has $m-1$ vertices. By induction, $\Gamma_1$ lies on the $[0,(m-1)-1]\times [4W-2(m-1)+2,4W]$ grid. Hence, $\Gamma'_1$ lies on the $[1,m-1]\times [4W-2m+2,4W-2]$ grid. Since $r(T_1)$ is placed at $(0,4W)$, we get that $\Gamma$ lies on the $[0,m-1]\times [4W-2m+2,4W]$~grid.

\cref{prop:bottom-to-top-trees} of \cref{def:quadrant-drawings} is satisfied by $\Gamma$ given that $k=1$ and that $r(T_1)$ is placed at $(0,4W)$. 

\cref{prop:bottom-to-top-sub-trees} of \cref{def:quadrant-drawings} is trivially true given that $\mathcal F$ consists of one tree only.

Next, we prove that $\Gamma$ satisfies~\cref{prop:bottom-to-top-internal} of \cref{def:quadrant-drawings}. Every vertex $v$ of $\mathcal F_1$ satisfies the requirements of the condition, given that $\Gamma'_1$ is a translation of $\Gamma_1$ and given that $\Gamma_1$ satisfies~\cref{prop:bottom-to-top-internal} of \cref{def:quadrant-drawings}. Further, $r(T_1)$ also satisfies the requirements of the condition; namely, the vertices of $T_1(u_j)$ have $y$-coordinates strictly smaller than the vertices of $T_1(u_{j+1})$, for $j=1,\dots,\ell_1-1$, given that $\Gamma'_1$ is a translation of $\Gamma_1$ and given that $\Gamma_1$ satisfies \cref{prop:bottom-to-top-sub-trees} of \cref{def:quadrant-drawings}. 

Finally, concerning~\cref{prop:wedge-free-internally} of \cref{def:quadrant-drawings}, we have that every vertex $v$ of $\mathcal F_1$ satisfies the requirements of the condition, given that $\Gamma'_1$ is a translation of $\Gamma_1$ and given that $\Gamma_1$ satisfies~\cref{prop:wedge-free-internally} of \cref{def:quadrant-drawings}. It remains to prove that $S_{\sewedge}(r(T_1))$ does not intersect $\Gamma$ other than along $\ell_{\oursearrow}(r(T_1))$. The part of $S_{\sewedge}(r(T_1))$ below or on the line $y:=4W-2$ coincides with $S_{\sewedge}\big(r(T(u_{\ell_1}))\big)$, given that $r(T(u_{\ell_1}))$ lies at the intersection between $\ell_{\oursearrow}(r(T_1))$ and $y:=4W-2$; hence, this part of $S_{\sewedge}(r(T_1))$ does not intersect $\Gamma$ other than along $\ell_{\oursearrow}(r(T_1))$, by induction. The part of $S_{\sewedge}(r(T_1))$ above the line $y:=4W-2$ contains only $r(T_1)$, by construction.


In the second inductive case we have $k>1$. Recall that $\Gamma$ has been obtained from the 
inductively constructed $\sedrawing$-drawings $\Gamma_1$ of $T_1$ and $\Gamma_2$ of $\mathcal F_2 = T_2,T_3,\dots,T_{k}$, by translating $\Gamma_1$ vertically down in such a way that $r(T_1)$ lies two units below the lowest vertex in $\Gamma_2$. Let $\Gamma'_1$ be the drawing of $T_1$ in $\Gamma$.

First, the planarity of $\Gamma$ follows from the planarity of $\Gamma_1$ and $\Gamma_2$, from the fact that $\Gamma'_1$ is a translation of $\Gamma_1$, and from the fact $\Gamma'_1$ and $\Gamma_2$ are separated by a horizontal line.

Second, $\Gamma$ is strictly-upward, strictly-leftward, and order-preserving since so are $\Gamma_1$ and $\Gamma_2$, and since $\Gamma'_1$ is a translation of $\Gamma_1$.

Third, $\Gamma$ is a grid drawing given that $\Gamma_1$ and $\Gamma_2$ are grid drawings and given that $\Gamma'_1$ is a translation by an integer vector of $\Gamma_1$.

We now prove that $\Gamma$ satisfies~\cref{prop:grid} of \cref{def:quadrant-drawings}. Let $m_1$ and $m_2$ be the number of vertices of $T_1$ and $\mathcal F_2$, respectively. By induction, $\Gamma_1$ lies on the $[0,m_1-1]\times [4W-2m_1+2,4W]$ grid and $\Gamma_2$ lies on the $[0,m_2-1]\times [4W-2m_2+2,4W]$ grid. Hence, $\Gamma'_1$ lies on the $[0,m_1-1]\times [4W-2m_1-2m_2+2,4W-2m_2]$ grid. Thus, $\Gamma$ lies on the $[0,\max\{m_1,m_2\}-1]\times  [4W-2m_1-2m_2+2,4W]$ grid. Since $\max\{m_1,m_2\}<m$ and $m_1+m_2=m$, we have that $\Gamma$ lies on the $[0,m-1]\times [4W-2m+2,4W]$ grid.

\cref{prop:bottom-to-top-trees,prop:bottom-to-top-sub-trees} of \cref{def:quadrant-drawings} are satisfied by $\Gamma$ since they are satisfied by $\Gamma_1$ and $\Gamma_2$, and since, by construction, the drawing $\Gamma_1$ is translated vertically down in such a way that $r(T_1)$ (and thus every vertex of $T_1$) lies below every vertex of $T_2$. 

\cref{prop:bottom-to-top-internal} of \cref{def:quadrant-drawings} is satisfied by $\Gamma$ since it is satisfied by each of $\Gamma_1$ and $\Gamma_2$ and since $\Gamma'_1$ is a translation of $\Gamma_1$.

Finally, concerning~\cref{prop:wedge-free-internally} of \cref{def:quadrant-drawings}, consider any vertex $v$ of $\mathcal F$. If $v$ is a vertex of $T_1$, then $S_{\sewedge}(v)$ does not intersect $\Gamma'_1$ other than along $\ell_{\oursearrow}(v)$, given that $\Gamma_1$ satisfies~\cref{prop:wedge-free-internally} of \cref{def:quadrant-drawings} and given that $\Gamma'_1$ is a translation of $\Gamma_1$; further,  $S_{\sewedge}(v)$ does not intersect $\Gamma_2$, since $\Gamma_2$ lies above the horizontal line through $r(T_1)$, while $S_{\sewedge}(v)$ lies below or on the same line. If $v$ is a vertex of $\mathcal F_2$, then $S_{\sewedge}(v)$ does not intersect $\Gamma_2$ other than along $\ell_{\oursearrow}(v)$, given that $\Gamma_2$ satisfies~\cref{prop:wedge-free-internally} of \cref{def:quadrant-drawings}; it remains to prove that $S_{\sewedge}(v)$ does not intersect $\Gamma'_1$. By \cref{prop:drawing-inside-triangle}, every vertex $v$ of $T_1$ lies inside or on the boundary of the triangle delimited by the $y$-axis, by $\ell_{\oursearrow}(r(T_1))$, and by the horizontal line $y:=2W+2$. Since $v$ lies above $r(T_1)$, we have that $\ell_{\oursearrow}(v)$ lies above $\ell_{\oursearrow}(r(T_1))$. This implies that $S_{\sewedge}(v)$ does not intersect $\Gamma'_1$. 
%
\end{proof}

Algorithms $\swdrawing$-{\sc Drawer}, $\nwdrawing$-{\sc Drawer}, and $\nedrawing$-{\sc Drawer} that construct a $\swdrawing$-drawing, a $\nwdrawing$-drawing, and a $\nedrawing$-drawing of $\mathcal F$ can be defined symmetrically to the algorithm $\sedrawing$-{\sc Drawer}.

We now go back to the problem of drawing a tree story $\langle T, \tau, W \rangle$. Let $n=|V(T)|$. Recall that $V(T)$ is partitioned into buckets $B_1,\dots,B_h$, where $h=\lceil \frac{n}{W} \rceil$. We now show how to partition the subtrees of $T$ induced by each bucket into up to two rooted ordered forests, so that the algorithms $\sedrawing$-, $\swdrawing$-, $\nwdrawing$-, and $\nedrawing$-{\sc Drawer} can be exploited in order to draw such forests, thus obtaining a drawing story of $\langle T, \tau, W \rangle$. We proceed in several phases as follows.


\paragraph{\underline{\sc Phase 1:}} We label each vertex $v$ of $T$ belonging to a bucket $B_i$ with the label $b(v)=i$ and we remove from $T$ all the edges $(u,v)$ such that $|b(u)-b(v)|>1$. Observe that such edges are not visualized in a drawing story of $\langle T, \tau, W \rangle$.

\paragraph{\underline{\sc Phase 2:}} As $T$ might have been turned into a forest by the previous edge removal, we add dummy edges to $T$ to turn it back into a tree, while ensuring that $|b(u)-b(v)|\leq 1$ for every edge $(u,v)$ of $T$. This is possible due to the following.

\begin{lemma} \label{le:edge-addition}
Dummy edges can be added to $T$ in total $O(n)$ time so that $T$ becomes a tree and every edge $(u,v)$ of $T$ is such that $|b(u)-b(v)|\leq 1$. 
\end{lemma}

\begin{proof}
First, we compute the connected components of $T$ and we label each vertex with the connected component it belongs to; this is done in $O(n)$ time. Now we arbitrarily pick a connected component $T^*$ of $T$. Note that the vertices of $T^*$ belong to a set of consecutive buckets $B_i,B_{i+1},\dots,B_j$. We visit $T^*$ in $O(|V(T^*)|)$ time to elect a {\em representative vertex} of $T^*$ for each of $B_i,B_{i+1},\dots,B_j$. Namely, when we visit a vertex $u$, we elect it as the representative vertex of $T^*$ for $B_{b(u)}$ if no representative vertex  of $T^*$ for $B_{b(u)}$ has been elected yet. 

Now the algorithm proceeds in steps, while maintaining the invariant that every edge $(u,v)$ of $T$ is such that $|b(u)-b(v)|\leq 1$. In a single step, we add an edge between a representative vertex of $T^*$ and a vertex in a  connected component of $T$ different from $T^*$. This is done as follows. Suppose that the vertices of $T^*$ belong to a set of consecutive buckets $B_p,B_{p+1},\dots,B_q$ such that $p>1$. Then we consider any vertex $v$ of $B_{p-1}$; let $T_v$ be the connected component of $T$ containing $v$ and note that $T_v\neq T^*$. We visit $T_v$ in $O(|V(T_v)|)$ time; when we visit a vertex $u$, we elect it as the representative vertex of $T^*$ for the bucket $B_{b(u)}$ if no representative vertex  of $T^*$ for $B_{b(u)}$ has been elected yet. Then we add to $T$ the edge between $u$ and the representative vertex $v$ of $T^*$ for $B_p$; note that $b(v)-b(u)= 1$. Now $T^*$ also contains $T_v$, as well as the edge $(u,v)$. A connected component $T_v$ can be analogously merged with $T^*$ in $O(|V(T_v)|)$ time if $q<h$. Eventually, $T^*$ contains a representative vertex for each bucket $B_1,B_2,\dots,B_h$. We consider the remaining connected components of $T$ different from $T^*$ one at a time. For each such a component, we select any vertex $u$, which we connect to the representative vertex $v$ of $T^*$ for $B_{b(u)}$ in $O(1)$ time; note that $b(v)-b(u)= 0$.\end{proof}

\paragraph{\underline{\sc Phase 3:}} We now root $T$ at an arbitrary vertex $r(T)$ in $B_1$. A \emph{pertinent component} of $T$ is a maximal connected component of $T$ composed of vertices in the same bucket. We assign a \emph{label} $b(P)=i$ to a pertinent component $P$ if every vertex of $P$ belongs to $B_i$. We now construct the following sets $R_1,R_2,\dots,R_k$ of pertinent components of $T$; see \cref{fi:treesExample}. 
The set $R_1$ only contains the pertinent component of $T$ the vertex $r(T)$ belongs to.
For $j>1$, the set $R_j$ contains every pertinent component~$P$ of $T$ such that
\begin{inparaenum}[(i)]
\item $P$ does not belong to $\bigcup^{j-1}_{i=1} R_i$ and
\item $P$ contains a vertex that is adjacent to a vertex belonging to a pertinent component in $R_{j-1}$.
\end{inparaenum}

By the construction of the $R_j$'s, since $|b(u)-b(v)|\leq 1$ for every edge $(u,v)$ of $T$, and by the rooting of $T$, we have the following simple property.

\begin{property} \label{prop:children}
	For every vertex $v \in R_j$, each child of $v$ belongs to either $R_j$~or~$R_{j+1}$.
\end{property}

%

\paragraph{\underline{\sc Phase 4:}} Next, we turn $T$ into an ordered tree by specifying a left-to-right order for the children of each vertex $v$. 
Let $R_j$ be the set $v$ belongs to. Then, by \cref{prop:children}, each child of $v$ is either in $R_j$ or in $R_{j+1}$. 
We set the left-to-right order of the children of $v$ so that all those in $R_j$ come first, in any order, and all those in $R_{j+1}$ come next, in any order.

\paragraph{\underline{\sc Phase 5:}} We are going to define rooted ordered forests as follows.
For $i=1,\dots,h$ with $i$ odd, we define:
\begin{enumerate}
	\item $\mathcal F_{i \cdot {\sedrawing}}$ as the forest containing all the pertinent components $P$ such that $b(P)=i$ and such that $P$ belongs to a set $R_j$ with $j \equiv 1 \mod 4$, and
	\item $\mathcal F_{i \cdot {\nwdrawing}}$ as the forest containing all the pertinent components $P$ such that $b(P)=i$ and such that $P$ belongs to a set $R_j$ with $j \equiv 3 \mod 4$.
\end{enumerate}
Also, for $i=1,\dots,h$ with $i$ even, we define:
\begin{enumerate}
	\item $\mathcal F_{i \cdot {\swdrawing}}$ as the forest containing all the pertinent components $P$ such that $b(P)=i$ and such that $P$ belongs to a set $R_j$ with $j \equiv 2 \mod 4$, and
	\item $\mathcal F_{i \cdot {\nedrawing}}$ as the forest containing all the pertinent components $P$ such that $b(P)=i$ and such that $P$ belongs to a set $R_j$ with $j \equiv 0 \mod 4$.
\end{enumerate}

We have the following. 

\begin{observation}\label{obs:edge-assignment}
Let $v$ be a vertex of $T$ and $u$ be its parent. Let $R_i$ and $R_j$ be the sets containing the pertinent components $u$ and $v$ belong to, respectively, where $j=i$ or $j = i+1$, by~\cref{prop:children}. Then the following cases are possible.
\begin{enumerate}[\bf \ref{obs:edge-assignment}a]
\item \label{obs:a} If $j=i$, then $u$ and $v$ both belong to either $\mathcal F_{i \cdot {\sedrawing}}$, $\mathcal F_{i \cdot {\nwdrawing}}$, $\mathcal F_{i \cdot {\swdrawing}}$, or $\mathcal F_{i \cdot {\nedrawing}}$.
\item \label{obs:b} If $j=i+1$, then $v$ is the root of a pertinent component in $R_j$. Also, either: 
\begin{enumerate}[(i)]
\item $i$ is odd, $u \in \mathcal F_{i \cdot {\sedrawing}}$, and $v \in \mathcal F_{i+1 \cdot {\swdrawing}}$; 
\item $i$ is even, $u \in \mathcal F_{i \cdot {\swdrawing}}$, and $v \in \mathcal F_{i+1 \cdot {\nwdrawing}}$; 
\item $i$ is odd, $u \in \mathcal F_{i \cdot {\nwdrawing}}$, and $v \in \mathcal F_{i+1 \cdot {\nedrawing}}$; or
\item $i$ is even, $u \in \mathcal F_{i \cdot {\nedrawing}}$, and $v \in \mathcal F_{i+1 \cdot {\sedrawing}}$.
\end{enumerate}
\end{enumerate}
\end{observation}

For each pertinent component $P$ in any $\mathcal F_{i \cdot X}$, with $i\in \{1,\dots,h\}$ and $X \in \{\sedrawing,\swdrawing, \nwdrawing, \nedrawing\}$, let $R_j$ be the set $P$ belongs to. If $j=1$, then the root of $P$ is $r(T)$, otherwise the root of $P$ is the vertex of $P$ that is adjacent to a vertex in $R_{j-1}$. Further, the left-to-right order of the children of every vertex of $P$ is the one inherited from $T$.
Finally, the linear ordering of the pertinent components in $\mathcal F_{i \cdot X}$ is defined as follows.
Let $P_1$ and $P_2$ be any two pertinent components in $\mathcal F_{i \cdot X}$ and let $R_j$ and $R_k$ be the sets containing $P_1$ and $P_2$, respectively. If $j<k$, then $P_1$ precedes $P_2$ in $\mathcal F_{i \cdot X}$. If $j>k$, then $P_1$ follows $P_2$ in $\mathcal F_{i \cdot X}$. Otherwise $j=k$; let $x$ be the lowest common ancestor of the roots of $P_1$ and of $P_2$ in $T$. Also, let $p_1$ and $p_2$ be the paths connecting the roots of $P_1$ and of $P_2$ with $x$, respectively. Further, let $x_1$ and $x_2$ be the children of $x$ belonging to $p_1$ and to $p_2$, respectively. Then, $P_1$ precedes $P_2$ in $\mathcal F_{i \cdot X}$ if and only if $x_1$ precedes $x_2$ in the left-to-right order of the children of $x$. We have the following.

\begin{lemma} \label{le:construction-forests}
Given the sets $R_1,\dots,R_k$, the rooted ordered forests $\mathcal F_{i \cdot X}$, with $i=1,\dots,h$ and $X \in \{\sedrawing,\swdrawing, \nwdrawing, \nedrawing\}$, can be computed in total $O(n)$ time.
\end{lemma}

\begin{proof}
First, we initialize new lists $S_1,\dots,S_k$, as well as the rooted ordered forests $\mathcal F_{i\cdot X}$, to empty lists. Second, we perform in $O(n)$ time a pre-order traversal of $T$ such that the children of each vertex are visited according to their left-to-right order (recall that $T$ is a rooted ordered tree). Each time we encounter the root of a pertinent component $P\in R_j$, we append $P$ in $O(1)$ time to~$S_j$. Finally, we scan the lists $S_1,\dots,S_k$ in this order. During the scan of a list $S_j$, we append the pertinent component $P\in R_j$ rooted at the currently scanned root to the corresponding forest $\mathcal F_{i\cdot X}$; such a forest is determined in $O(1)$ time by looking at the value of $b(P)$ (this determines $i$) and by computing $j\mod 4$ (this determines $X$). 
\end{proof}


We are now ready to state the following main result.

%

\begin{theorem}\label{th:trees}
For any tree story $\langle T, \tau, W \rangle$ such that $T$ has $n$ vertices, it is possible to construct in $O(n)$ time a drawing story that is planar, straight-line, and lies on an~$(8W+1) \times (8W+1)$~grid.
\end{theorem}

\begin{figure}[h!]
\centering
\includegraphics[width=0.9\textwidth]{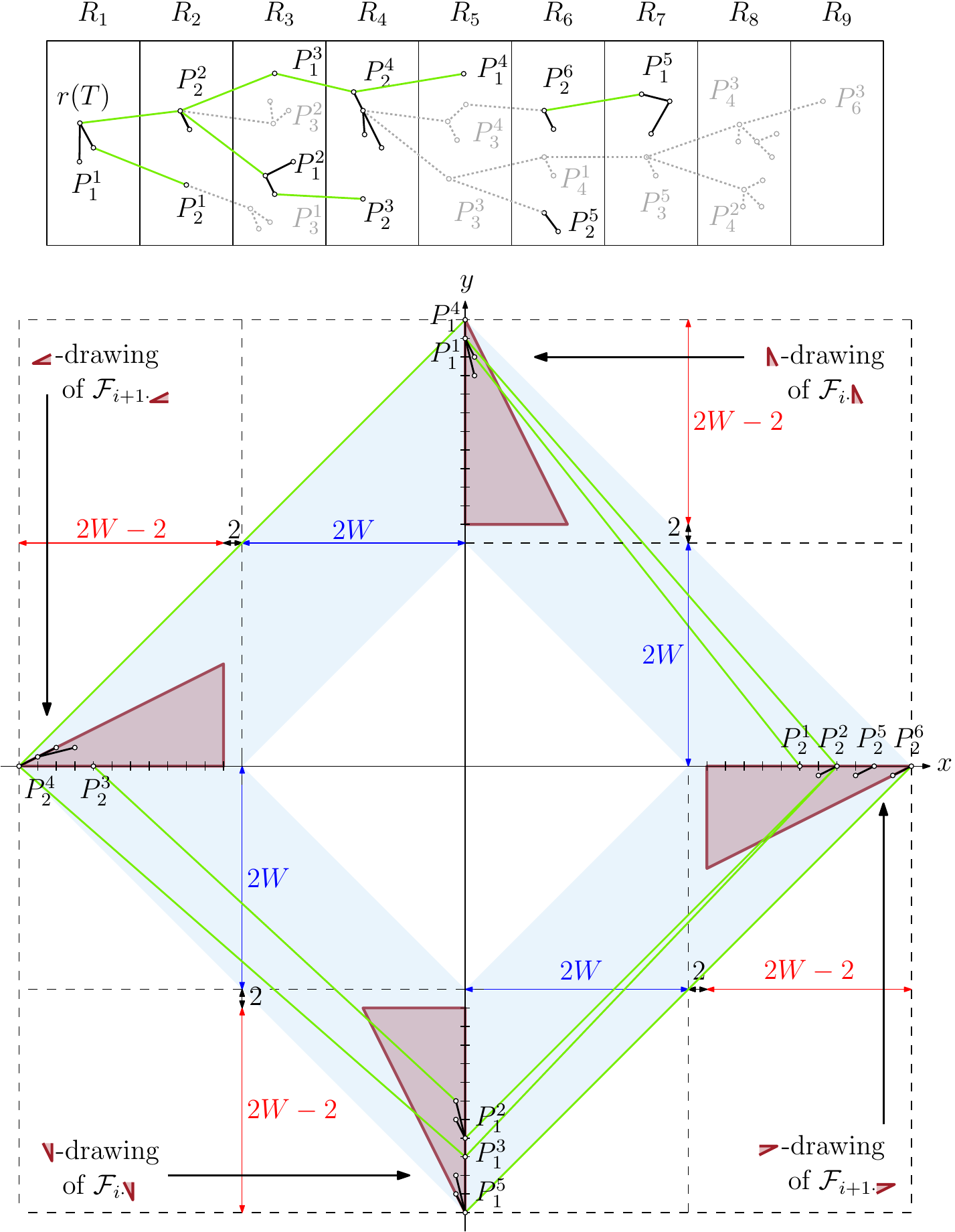}
\caption{Illustration for the proof of \cref{th:trees}, with $W=12$. The upper part of the figure shows the rooted ordered tree $T$; vertices and edges that are not in $T[B_{1,2}]$ are gray. A pertinent component $P^j_i$ of $T$ belongs to the bucket $B_i$; further, the index $j$ represents the order of the components in the corresponding rooted forests. The lower part of the figure shows the drawing $\Gamma[B_{1,2}]$ of $T[B_{1,2}]$ constructed by the algorithm.}\label{fi:treesExample}
\end{figure}

\begin{proof}
We construct a planar straight-line drawing story $\Gamma$ of $\langle T, \tau, W \rangle$. We perform {\sc Phases 1}--{\sc 5} in order to construct the rooted ordered forests $\mathcal F_{i\cdot X}$ with $i\in \{1,\dots,h\}$ and $X \in \{\protect\sedrawing,\protect\swdrawing, \protect\nwdrawing, \protect\nedrawing\}$. Note that {\sc Phase 2} introduces in $T$ some dummy edges; after the construction of $\Gamma$, such edges are removed from the actual drawing story of $\langle T, \tau, W \rangle$. Since $\Gamma$ is a straight-line drawing, it suffices to describe how to assign coordinates to the vertices of $T$; see \cref{fi:treesExample}. For each rooted ordered forest $\mathcal F_{i\cdot X}$, we apply the algorithm $X$-{\sc Drawer} in order to construct an $X$-drawing $\Gamma_{i\cdot X}$ of  $\mathcal F_{i\cdot X}$. We let the coordinates of each vertex $v$ of $\mathcal F_{i\cdot X}$ in $\Gamma$ coincide with the coordinates of $v$ in $\Gamma_{i\cdot X}$.

%


We now prove that $\Gamma$ satisfies the properties in the statement of the theorem. By \cref{prop:grid} of \cref{def:quadrant-drawings}, a $\sedrawing$-drawing lies on the $[0,4W] \times [0,4W]$ grid. Similarly, a $\swdrawing$-, a $\nwdrawing$-, and a $\nedrawing$-drawing lies on the $[0,4W] \times [-4W,0]$ grid, on the $[-4W,0] \times [-4W,0]$ grid, and on the $[-4W,0] \times [0,4W]$ grid, respectively. Thus, $\Gamma$ lies on the $[-4W,4W] \times [-4W,4W]$ grid. 

Concerning the running time for the construction of $\Gamma$, we have that the initial modification of $T$, which ensures that $|b(u)-b(v)|\leq 1$ for every edge $(u,v)$ of $T$, can be done in $O(n)$ time, by Lemma~\ref{le:edge-addition}. The labeling $b(u)$ of each vertex $u$ of $T$ is easily done in $O(n)$ time. The construction of the sets $R_1,\dots,R_k$ can be accomplished by an $O(n)$-time traversal of $T$ starting from $r(T)$. The construction of the rooted ordered forests $\mathcal F_{i\cdot X}$, with $i=1,\dots,h$ and with $X \in \{\sedrawing,\swdrawing, \nwdrawing, \nedrawing\}$, can be performed in total $O(n)$ time by \cref{le:construction-forests}. Finally, by \cref{le:quadrant-drawings}, the algorithm $X$-{\sc Drawer} runs in linear time in the size of its input $\mathcal F_{i\cdot X}$.

It remains to show that each drawing $\Gamma_t \in \Gamma$ is planar. We exploit the following lemma. Let $B_{i,i+1}=B_i \cup B_{i+1}$. 

\begin{lemma}  \label{cl:planarity}
For $i=1,\dots,h-1$, \mbox{the drawing $\Gamma[B_{i,i+1}]$ of $T[B_{i,i+1}]$ is planar.}
\end{lemma}

\begin{proof}
Suppose that $i$ is odd; the case in which $i$ is even can be treated analogously.
Let $(u,v)$ and $(w,z)$ be any two edges of $T[B_{i,i+1}]$. We prove that $(u,v)$ and $(w,z)$ do not cross in $\Gamma[B_{i,i+1}]$. Analogous and simpler arguments prove that no edge overlaps with a vertex in $\Gamma[B_{i,i+1}]$.

We distinguish four cases.
\begin{description}
\item[\bf {Case} (i):] $\{u,v\} \subseteq B_i$ or $\{u,v\} \subseteq B_{i+1}$, and $\{w,z\} \subseteq B_i$ or $\{w,z\} \subseteq B_{i+1}$,
\item[\bf {Case} (ii):] $\{u,v,w\} \subseteq B_i$ and $z \in B_{i+1}$, 
\item[\bf {Case} (iii):] $u \in B_i$ and $\{v,w,z\} \subseteq B_{i+1}$, and
\item[\bf {Case} (iv):] $\{u,w\} \subseteq B_i$ and $\{v,z\} \subseteq B_{i+1}$.
\end{description}

\paragraph{{Case}~(i).} \cref{le:quadrant-drawings} ensures the planarity of the drawings $\Gamma_{i \cdot \sedrawing}$, $\Gamma_{i+1 \cdot \swdrawing}$, $\Gamma_{i \cdot \nwdrawing}$, and $\Gamma_{i+1 \cdot \nedrawing}$. Further, \cref{prop:grid} of \cref{def:quadrant-drawings} ensures that such drawings do not overlap with each other. It follows that the drawing of $\mathcal F_{i \cdot \sedrawing} \cup \mathcal F_{i+1 \cdot \swdrawing} \cup \mathcal F_{i \cdot  \nwdrawing} \cup \mathcal F_{i+1 \cdot \nedrawing}$ in $\Gamma[B_{i,i+1}]$ is planar. By \cref{obs:edge-assignment}\darkblue{a}, we have that $(u,v)$ and $(w,z)$ belong to $\mathcal F_{i \cdot \sedrawing} \cup \mathcal F_{i+1 \cdot \swdrawing} \cup \mathcal F_{i \cdot  \nwdrawing} \cup \mathcal F_{i+1 \cdot \nedrawing}$, hence they do not cross each other in $\Gamma[B_{i,i+1}]$.

%
%
\paragraph{{Case}~(ii).} If $\{u,v,w\} \subseteq B_i$, then $u$, $v$, and $w$ belong to $\mathcal F_{i \cdot \sedrawing}\cup \mathcal F_{i \cdot  \nwdrawing}$. Suppose that the edge $(u,v)$ belongs to $\mathcal F_{i \cdot \sedrawing}$, the case in which it belongs to $\mathcal F_{i \cdot  \nwdrawing}$ is symmetric. 

Suppose first that $w$ belongs to $\mathcal F_{i \cdot  \nwdrawing}$. Then, by \cref{obs:edge-assignment}\darkblue{b}, we have that either $z$ is the root of a pertinent component in $\mathcal F_{i+1 \cdot  \nedrawing}$ or $z$ belongs to $\mathcal F_{i+1 \cdot  \swdrawing}$. In either case, by \cref{prop:grid} of \cref{def:quadrant-drawings}, we have that $\Gamma_{i \cdot \sedrawing}$ (and in particular the edge $(u,v)$) lies above the $x$-axis, while by \cref{prop:bottom-to-top-trees,prop:grid} of \cref{def:quadrant-drawings} the edge $(w,z)$ lies below or on the $x$-axis. Thus, $(u,v)$ and $(w,z)$ do not cross each other  in $\Gamma[B_{i,i+1}]$. 

Suppose next that $w$ belongs to $\mathcal F_{i \cdot  \sedrawing}$. By \cref{obs:edge-assignment}\darkblue{b}, we have that either $z$ is the root of a pertinent component in $\mathcal F_{i+1 \cdot  \swdrawing}$, or $z$ belongs to $\mathcal F_{i+1 \cdot  \nedrawing}$ and $w$ is the root of a pertinent component in $\mathcal F_{i \cdot  \sedrawing}$. In the latter case, $(u,v)$ lies to the right of the $y$-axis, except possibly for one of its end-vertices which might be on the $y$-axis, while $(w,z)$ lies to the left of the $y$-axis, except for $w$ which lies on the $y$-axis. By \cref{prop:bottom-to-top-trees} of \cref{def:quadrant-drawings}, the edges $(u,v)$ and $(w,z)$ do not cross each other  in $\Gamma[B_{i,i+1}]$. In the former case, by \cref{prop:wedge-free-internally} of \cref{def:quadrant-drawings}, we have that the interior of the wedge $S_{\sewedge}(w)$ does not intersect $\Gamma_{i \cdot  \sedrawing}$; further, by \cref{prop:containment}, the segment $\big[(2W+2,0),(4W,0)\big]$ lies in the interior of the wedge $S_{\sewedge}(w)$. Since, by \cref{prop:bottom-to-top-trees} of \cref{def:quadrant-drawings}, the vertex $z$ lies on the segment $\big[(2W+2,0),(4W,0)\big]$, it follows that the edge $(w,z)$ lies in the interior of $S_{\sewedge}(w)$, except for $w$. Thus, $(u,v)$ and $(w,z)$ do not cross each other  in $\Gamma[B_{i,i+1}]$. 

\paragraph{{Case}~(iii).} The treatment of this case is symmetric to the one of {\bf {Case}~(ii)}.

\paragraph{{Case}~(iv).} Since $\{u,w\} \subseteq B_i$, then either $u$ and $w$ belong to the same rooted ordered forest $\mathcal F_{i \cdot \sedrawing}$ or $\mathcal F_{i \cdot  \nwdrawing}$, or they belong one to  $\mathcal F_{i \cdot \sedrawing}$ and one to  $\mathcal F_{i \cdot  \nwdrawing}$.

In the latter case, assume that $u$ belongs to  $\mathcal F_{i \cdot \sedrawing}$ and $w$ to $\mathcal F_{i \cdot  \nwdrawing}$. By \cref{obs:edge-assignment}\darkblue{b}, each of $v$ and $z$ belongs to either $\mathcal F_{i+1 \cdot \swdrawing}$ or $\mathcal F_{i+1 \cdot  \nedrawing}$. If $v$ belongs to $\mathcal F_{i+1 \cdot \swdrawing}$ and $z$ to $\mathcal F_{i+1 \cdot  \nedrawing}$, by \cref{prop:grid} of \cref{def:quadrant-drawings} we have that $(u,v)$ lies to the right of the $y$-axis, except possibly for $u$ which might lie on the positive $y$-axis, while $(w,z)$ lies to the left of the $y$-axis, except possibly for $w$ which might lie on the negative $y$-axis. Thus, $(u,v)$ and $(w,z)$ do not cross each other  in $\Gamma[B_{i,i+1}]$. If $v$ belongs to $\mathcal F_{i+1 \cdot \nedrawing}$ and $z$ to $\mathcal F_{i+1 \cdot  \swdrawing}$, by \cref{prop:grid} of \cref{def:quadrant-drawings} we have that $(u,v)$ lies above the $x$-axis, except possibly for $v$ which might lie on the negative $x$-axis, while $(w,z)$ lies below the $x$-axis, except possibly for $z$ which might lie on the positive $x$-axis. Thus, $(u,v)$ and $(w,z)$ do not cross each other  in $\Gamma[B_{i,i+1}]$. If $v$ and $z$ both belong to $\mathcal F_{i+1 \cdot \swdrawing}$, then by \cref{obs:edge-assignment}\darkblue{b} we have that $v$ is the root of a pertinent component of $\mathcal F_{i+1 \cdot \swdrawing}$; hence by \cref{prop:grid,prop:bottom-to-top-trees} of \cref{def:quadrant-drawings} we have that $(u,v)$ lies above the $x$-axis, except for $v$ which lies on the $x$-axis, while $(w,z)$ lies below the $x$-axis, except possibly for $z$ which might lie on the $x$-axis. Thus, $(u,v)$ and $(w,z)$ do not cross each other  in $\Gamma[B_{i,i+1}]$. The case in which $v$ and $z$ both belong to $\mathcal F_{i+1 \cdot \nedrawing}$ is symmetric to the previous one.

In the former case, assume that $u$ and $w$ both belong to $\mathcal F_{i \cdot \sedrawing}$. We distinguish four cases: (a) $v$ belongs to $\mathcal F_{i+1 \cdot \nedrawing}$ and $z$ belongs to $\mathcal F_{i+1 \cdot \swdrawing}$; (b) $v$ belongs to $\mathcal F_{i+1 \cdot \swdrawing}$ and $z$ belongs to $\mathcal F_{i+1 \cdot \nedrawing}$; (c) $v$ and $z$ both belong to $\mathcal F_{i+1 \cdot \swdrawing}$; and (d) $v$ and $z$ both belong to $\mathcal F_{i+1 \cdot \nedrawing}$. Cases~(a) and~(b) are symmetric to the already discussed case in which $u$ belongs to  $\mathcal F_{i \cdot \sedrawing}$, $w$ belongs to $\mathcal F_{i \cdot  \nwdrawing}$, and $v$ and $z$ belong to $\mathcal F_{i+1 \cdot \swdrawing}$. Further, case (d) is symmetric to case (c), which we discuss next.

%
%
By \cref{obs:edge-assignment}\darkblue{b}, we have that $v$ and $z$ are the roots of two pertinent components $P_v$ and $P_z$ of $\mathcal F_{i+1 \cdot \swdrawing}$, respectively. By \cref{prop:bottom-to-top-trees} of \cref{def:quadrant-drawings}, the vertices $v$ and $z$ lie on the segment $\big[(2W+2,0),(4W,0)\big]$. Assume w.l.o.g.\ that $z$ lies to the right of $v$. We have the following.

\begin{claim} \label{cl:w-above-u}
The vertex $w$ lies above the vertex $u$ in $\Gamma[B_{i,i+1}]$. 
\end{claim}

\begin{proof}
Let $R_j$ and $R_k$ be the sets containing $v$ and $z$, respectively.
By the construction of the ordering of $\mathcal F_{i+1 \cdot \swdrawing}$ and since $z$ lies to the right of $v$, we have that two cases are possible.	
	
In the first case we have $j<k$. By \cref{obs:edge-assignment}\darkblue{b}, we have that $u$ and $w$ belong to $R_{j-1}$ and $R_{k-1}$, respectively. Since $j-1<k-1$, we have that $u$ and $w$ belong to distinct pertinent components $P_u$ and $P_w$ of $T$, respectively, where $P_u$ precedes $P_w$ in $\mathcal F_{i \cdot \sedrawing}$. By \cref{prop:bottom-to-top-sub-trees} of \cref{def:quadrant-drawings}, we have that $w$ lies above $u$ in $\Gamma[B_{i,i+1}]$.

In the second case we have $j=k$. Let $x$ be the lowest common ancestor of $v$ and $z$ in $T$. Also, let $v'$ and $z'$ be the children of $x$ on the paths connecting $x$ with $v$ and $z$, respectively. Then we have that $v'$ precedes $z'$ in the left-to-right order of the children of $x$. Note that $j-1=k-1$ and that $x$ is also the lowest common ancestor of $u$ and $w$ in $T$, where $u$ lies on the path between $v$ and $x$ and $w$ lies on the path between $z$ and $x$. 

\begin{itemize}
	\item If $u$ and $w$ belong to distinct pertinent components $P_u$ and $P_w$ of $T$ then, as in the first case, $P_u$ precedes $P_w$ in $\mathcal F_{i \cdot \sedrawing}$. By \cref{prop:bottom-to-top-sub-trees} of \cref{def:quadrant-drawings} we have that $w$ lies above $u$ in $\Gamma[B_{i,i+1}]$.
	\item If $u$ and $w$ belong to the same pertinent component of $T$ and neither of them is $x$, then by \cref{prop:bottom-to-top-internal} of \cref{def:quadrant-drawings} we have that $w$ lies above $u$ in $\Gamma[B_{i,i+1}]$.
	\item If $u$ and $w$ belong to the same pertinent component of $T$ and $w=x$, then we have that $w$ lies above $u$ in $\Gamma[B_{i,i+1}]$ since $\Gamma_{i \cdot \sedrawing}$ is strictly-upward.
	\item Finally, note that $u\neq x$. Indeed, suppose, for a contradiction, that $u$ is the lowest common ancestor of $v$ and $z$. By the choice of the left-to-right ordering of the children of $u$ (given in {\sc Phase 4}), we have that $z'$ precedes $v$ in this ordering. Therefore, by the construction of the ordering of $\mathcal F_{i+1 \cdot \swdrawing}$ (given in {\sc Phase 5}), we have that $P_z$ precedes $P_v$ in $\mathcal F_{i+1 \cdot \swdrawing}$; by \cref{prop:bottom-to-top-trees} of \cref{def:quadrant-drawings}, this contradicts the assumption that $z$ lies to the right of $v$.
\end{itemize}   	
	
This concludes the proof of the claim.
\end{proof}

By \cref{cl:w-above-u} and by \cref{prop:grid} of \cref{def:quadrant-drawings}, we have that $w$ lies above $u$, which lies above $v$ and $z$; recall that these last two vertices lie on the segment $\big[(2W+2,0),(4W,0)\big]$. Hence, the edge $(u,v)$ crosses the edge $(w,z)$ if and only if $u$ lies to the right of the edge $(w,z)$. By \cref{prop:containment}, the  edge $(w,z)$ lies in the wedge $S_{\sewedge}(w)$, except at $w$. However, if $u$ lies to the right of the edge $(w,z)$, then $S_{\sewedge}(w)$ contains $u$, which contradicts \cref{prop:wedge-free-internally} of \cref{def:quadrant-drawings}. 
\end{proof}


By \cref{prop:supergraph} and \cref{cl:planarity}, we have that $\Gamma$ is planar. This concludes the proof of the theorem.
\end{proof}

\section{Conclusions and Open Problems}\label{se:conclusions}

We have shown how to draw dynamic trees with straight-line edges, using an area that only depends on the number of vertices that are simultaneously present in the tree, while maintaining planarity. This result is obtained by partitioning the vertices of the tree into buckets and by establishing topological and geometric properties for the forests induced by pairs of consecutive buckets. Further, we proved that this result cannot be generalized to arbitrary planar graphs.  

Several interesting problems arise from this research. 

\begin{enumerate}
\item Do other notable families of planar graphs admit a planar straight-line drawing story on a grid of size polynomial in $W$, or polynomial in $W$ and sublinear in $n$? In particular, do outerplanar and series parallel graphs (which include trees) admit such representations?

For instance, it is not difficult to extend \cref{th:paths} to show that any graph story $\langle G, \tau, W \rangle$ such that $G$ is a cycle admits a drawing story that is planar, straight-line, and lies on a $(4W +2)\times (4W+2)$ grid.

\item Which bounds can be shown for a drawing story of a dynamic graph that is not a tree, while each graph of the story is a forest? Or for a drawing story of a dynamic graph that is not a path, while each graph of the story is a linear forest?

\item Which bounds can be shown on the area requirements of planar poly-line drawing stories of graphs both in terms of $W$ and the number $b$ of bends allowed on each edge?

\item How about extending our model to graphs where a vertex is allowed to appear several times? Or to graphs where multiple vertices might appear at the same time (and still at any moment the graph contains at most $W$ vertices)?

\item Can our results be extended to the setting in which edges, and not vertices, appear over time?
\end{enumerate}
%
%

%
%
%
\bibliographystyle{splncs04}
\bibliography{bibliography}
%




\end{document}